\documentclass[singlecolumn,aps,nofootinbib,superscriptaddress,tightenlines,11pt,notitlepage]{revtex4-1}
\usepackage{amsmath}
\usepackage{amsthm}
\usepackage{amsfonts}
\usepackage{color}
\usepackage{mathtools}

\usepackage[ colorlinks = true,
             linkcolor = blue,
             urlcolor  = blue,
             citecolor = red,
             anchorcolor = green,
]{hyperref}

\newtheorem{lemma}{Lemma}

\newtheorem{theorem}[lemma]{Theorem}
\newtheorem{corollary}[lemma]{Corollary}

\DeclareMathOperator{\tr}{tr}

\newcommand{\cS}{\mathcal{S}}

\newcommand{\cP}{\mathcal{P}}

\newcommand{\cB}{\mathcal{B}}

\newcommand{\eps}{\varepsilon}

\DeclareMathOperator{\argmin}{argmin}

\begin{document}

\title{A minimax approach to one-shot entropy inequalities}

\author{Anurag Anshu}
\affiliation{Institute for Quantum Computing, University of Waterloo, Waterloo, Canada}
\affiliation{Perimeter Institute for Theoretical Physics, Waterloo, Canada}

\author{Mario Berta}
\affiliation{Department of Computing, Imperial College London, England}

\author{Rahul Jain}
\affiliation{Center for Quantum Technologies, National University of Singapore and MajuLab, UMI 3654, Singapore}

\author{Marco Tomamichel}
\affiliation{Centre for Quantum Software and Information, University of Technology Sydney, Sydney}
\affiliation{Center for Quantum Technologies, National University of Singapore, Singapore}

\begin{abstract}
One-shot information theory entertains a plethora of entropic quantities, such as the smooth max-divergence, hypothesis testing divergence and information spectrum divergence, that characterize various operational tasks and are used to prove the asymptotic behavior of various tasks in quantum information theory. Tight inequalities between these quantities are thus of immediate interest. In this note we use a minimax approach (appearing previously for example in the proofs of the quantum substate theorem), to simplify the quantum problem to a commutative one, which allows us to derive such inequalities. Our derivations are conceptually different from previous arguments and in some cases lead to tighter relations. We hope that the approach discussed here can  lead to progress in open problems in quantum Shannon theory, and exemplify this by applying it to a simple case of the joint smoothing problem. 
\end{abstract}

\maketitle


\section{Introduction}
\label{sec:intro}

Recent years have seen remarkable progress in the area of one-shot quantum Shannon theory, which generalizes the standard asymptotic and i.i.d.\ (independent and identically distributed) quantum Shannon theory and also eases the notational complications in the latter. Achievability results in the one-shot setting clarify a lot about the structure of the protocol, as various entropic equalities that are equivalent in the asymptotic and i.i.d. setting are vastly different in the one shot setting. This setting also forces the development of novel encoding and decoding schemes that would have been trivial if the time sharing method was used in the asymptotic and i.i.d. setting.  

A (minor) downside of one-shot information theory is that there can be various quantities that seem to generalize the entropic quantities such as the relative entropy. Below, we introduce various such quantities that will be considered in this work. We focus here on relative entropies, but relations for other entropic quantities like entropy, conditional entropy and mutual information can often be derived readily using the fact that they can be expressed in terms of relative entropies. 

\subsection{Notation and definitions}

We will fix a finite-dimensional Hilbert space throughout most of this manuscript and denote with $\cP$ and $\cS$ the set of positive semi-definite operators and the subset of trace-normalized quantum states, respectively. Sometimes we will refer to the set of sub-normalized states, denoted $\cS_{\bullet}$, which contains all positive semi-definite operators $\rho \geq 0$ (using the L\"owner partial order) with $0 < \tr(\rho) \leq 1$. When joint quantum systems are considered, we use the notation $\cS(AB)$ etc.\ to denote joint quantum states on the Hilbert spaces $A$ and $B$.

Some of the entropic quantities will require the concept of a neighbourhood, namely a function $\cB$ that maps $\rho \in \cS$ to an $\eps$-neighbourhood $\cB^{\eps}(\rho) \subset \cS$ of $\rho$. We can also define neighbourhoods of sub-normalized states in the same way. We will always require that, for any $\rho \in \cS_{\bullet}$, the set $\cB^{\eps}(\rho)$ is convex and at least contains $\rho$. Such $\eps$-neighbourhoods can easily be constructed from any metric on states, and the two most prominent examples are defined below for any $\eps \in [0, 1)$. The first is the neighbourhood of states that are close in trace distance, $T(\rho,\sigma) := \frac12 \| \rho - \sigma \|$, given as
\begin{align}
   \cB^{\eps}_T(\rho) := \left\{ \tilde\rho \in \cS :  T(\rho, \tilde\rho) \leq \eps \right\} \,.
\end{align}
The second is the neighbourhood of sub-normalized states that are close in purified distance~\cite{tomamichel09},\begin{align}
   \cB^{\eps}_P(\rho) := \left\{ \tilde\rho \in \cS_{\bullet} : P(\rho,\tilde\rho) \leq \eps \right\} \,,
\end{align}
where $P(\rho,\sigma) = \sqrt{1 - \bar{F}(\rho, \sigma)}$ and
$\bar{F}(\rho,\sigma) = \big(\|\sqrt{\rho}\sqrt{\sigma}\|_1 + \sqrt{(1-\tr \rho)(1-\tr\sigma)} \big)^2$
is a generalization of the fidelity to sub-normalized states.

We are now ready to define our entropic quantities of interest. The max-divergence is defined for any $\rho \in \cS_{\bullet}$ and $\sigma \in \cP$ as
\begin{align}
   D_{\max}(\rho\|\sigma) := \inf \{ \lambda \in \mathbb{R} : \rho \leq 2^{\lambda} \sigma \} \,.
\end{align}
Note that by definition of the infimum this quantity takes on the value $+\infty$ in case there does not exist a $\lambda$ satisfying the constraint $\rho \leq 2^{\lambda} \sigma$, which happens if and only if the support of $\rho$ is not contained in the support of $\sigma$. Otherwise, the minimum is achieved and takes the value $\lambda^* = \log \| \sigma^{-\frac12} \rho \sigma^{-\frac12} \|_\infty$, where we used the Moore-Penrose inverse. Using any neighbourhood ball $\cB^{\eps}$, we define an $\eps$-smooth max-divergence as~\cite{renner05,datta08}
\begin{align}
   D_{\max}^{\eps}(\rho\|\sigma) := \inf_{\tilde\rho \in \cB^{\eps}(\rho)} D_{\max}(\tilde\rho\|\sigma) \,.
\end{align}
We will use the notation $D_{\max}^{\eps,P}$ and $D_{\max}^{\eps,T}$ to specify the balls $\cB_P^{\eps}$ and $\cB_T^{\eps}$, respectively.

The max-divergence is a limiting case of a R\'enyi divergence~\cite{lennert13,wilde13}, namely the family
\begin{align}
	\widetilde{D}_{\alpha}(\rho\|\sigma) :=\frac{1}{\alpha-1} \log \frac{ \tr \left( \sigma^{\frac{1-\alpha}{2\alpha}} \rho \sigma^{\frac{1-\alpha}{2\alpha}}\right)^{\alpha} } { \tr \rho } \,.
\end{align}
for $\alpha \in [\frac12, 1) \cup (1, \infty)$ defined for any $\rho \in \cS_{\bullet}$ and $\sigma \in \cP$. The max-relative divergence is recovered in the limit $\alpha \to \infty$ and the name is justified since the family is monotonically increasing as a function of $\alpha$. In the limit $\alpha \to 1$, we recover the relative entropy:
$$D(\rho\|\sigma):=  \frac{1}{\tr \rho} \tr \rho\left(\log\rho - \log\sigma \right) .$$

Asymmetric quantum hypothesis testing plays a crucial role in one-shot quantum information theory. The fundamental relationship between errors of the first and second kind can be cast as an entropic quantity.
Bounding the error of the first kind with $\eps \in [0, 1)$ and minimizing the error of the second kind, the $\eps$-hypothesis testing divergence is defined as
\begin{align}
	D_h^{\eps}(\rho\|\sigma) := -\log \sup_{0 \leq \Lambda \leq 1 \atop \tr \Lambda \rho \geq 1-\eps }  \tr \Lambda\sigma \,.
\end{align}

For any Hermitian operator $X$, let $\{X\}_+$ be the projector onto the subspace spanned by all the eigenvectors with positive eigenvalue. We define the $\eps$-information spectrum divergence as
\begin{align}
	D_s^{\eps}(\rho\|\sigma):= \sup\{ \lambda \in \mathbb{R} : \tr \rho\{2^\lambda\sigma - \rho\}_{+} \leq \eps\} \,.
\end{align}
This quantity gives a potential quantum generalization of the notion of $\eps$-tail bounds of the log-likelihood ratio function. To see this, note that for $P, Q$ two probability distributions, the above expression simplifies to
\begin{align}
	D_s^{\eps}(P \| Q):= \sup \left\{ \lambda \in \mathbb{R} : \Pr_P \bigg[ \log \frac{P}{Q} \leq \lambda \bigg] \leq \eps \right\} \,.
\end{align}
Its usefulness, apart from this simple interpretation, is mainly due to its close relation to hypothesis testing, shown in the following relation from~\cite[Lemma 12]{tomamichel12}: For any $\rho \in \cS$, $\sigma \in \cP$ and  $\eps, \delta\in(0,1)$ with $\eps+\delta<1$, it holds that
\begin{equation}
\label{eq:dsdh}
D_s^{\eps}(\rho\|\sigma) \leq D_h^{\eps}(\rho\|\sigma) \leq D_s^{\eps+\delta}(\rho\|\sigma) + \log\frac{1}{\delta} \,.
\end{equation}

\subsection{Some useful properties of above quantities}

The purified distance satisfies the following `gentle measurement' property, which has first been established in~\cite[Lemma 7]{tomamichel17b}. Since the relation between the below lemma and the result in~\cite{tomamichel17b} is not imediately obvious, we provide a proof in Appendix~\ref{app:gentle} for the convenience of the reader.
\begin{lemma}
For any projector $P$ and $\rho \in \cS_{\bullet}$, we have
\label{lem:purifiedgentle}
\begin{align}
	P(\rho, \tilde\rho) = \sqrt{\tr P\rho}  \qquad \textrm{for} \qquad \tilde\rho = \frac{(1-P)\rho(1-P)}{1 - \tr P \rho}  \,.
\end{align}
\end{lemma}
It is worth noting that the state $\tilde\rho$ is only normalized if $\rho \in \cS$ and sub-normalized otherwise. The special case of normalized $\rho$ is in fact well-known, and in that case we also have $T(\rho, \tilde\rho) \leq P(\rho,\tilde\rho) = \sqrt{\tr P\rho}$ by the Fuchs-van de Graaf inequality.

Many of these entropic quantities satisfy the data processing \cite{datta08,beigi13,frank13}. That is, for any quantum channel (a completely positive and trace-preserving map) $\mathcal{E}$, it holds that
\begin{equation}
\label{eq:dataproc}
D_h^{\eps}(\rho\|\sigma) \geq D_h^{\eps}(\mathcal{E}(\rho)\|\mathcal{E}(\sigma)), \quad \widetilde{D}_{\alpha}(\rho\|\sigma) \geq \widetilde{D}_{\alpha}(\mathcal{E}(\rho)\|\mathcal{E}(\sigma)), \quad D_{\max}^{\eps}(\rho\|\sigma) \geq D_{\max}^{\eps}(\mathcal{E}(\rho)\|\mathcal{E}(\sigma)).
\end{equation}
Data processing for the information spectrum divergence is not as simple, but an approximate data-processing inequality can be deduced from~\eqref{eq:dsdh}. Thus, information spectrum divergence is known to satisfy data processing only up to an additive logarithmic term.



\section{Relating various information theoretic measures}

Our central idea is inspired by the works \cite{JainRS02, Jain:2009, JainN12} on the quantum substate theorem, which show that we can use a minimax approach to find the optimal smoothing of the max-divergence. More precisely, we use the following straight-forward generalization of a key result from~\cite{JainN12}, a proof of which is given in Appendix~\ref{app:minimax} for the convenience of the reader.

\begin{lemma}\label{lem:minimax-smoothing}
Let $\rho \in \cS_{\bullet}$, $\sigma \in \cP$. For any convex $\eps$-neighbourhood $\cB^{\eps}(\rho)$, we have
\begin{align}
D_{\max}^{\eps}(\rho\|\sigma)=\sup_{M\geq0\atop \mathrm{Tr}[M\sigma]\leq1}\inf_{\tilde{\rho}\in \cB^{\eps}(\rho)}\log\mathrm{Tr}\left[M\tilde{\rho}\right].
\end{align}
\end{lemma}

\subsection{Smooth max-divergence and R{\'e}nyi relative entropies}

Our first application is a relation between $\eps$-smoth max-divergence and R\'enyi divergence, which improves on~\cite[Proposition 6.5]{mythesis} for the purified distance smoothing (which was shown using a different method) and is new for normalized trace distance smoothing. Our proof closely follows the proof of the quantum substate theorem in~\cite{JainN12}.

\begin{theorem}\label{thm:renyi}
Let $\rho \in \cS_{\bullet}$, $\sigma \in \cP$. For any $\eps \in (0,1)$ and $\alpha>1$, we have
\begin{align}
D_{\max}^{\eps,P}(\rho\|\sigma) &\leq \widetilde{D}_\alpha(\rho\|\sigma) + \frac{ 1}{\alpha-1} \log \frac{1}{\eps^2}  + \log \frac{1}{1-\eps^2} \,.
\end{align}
The same inequality also holds with $D_{\max}^{\eps,P}$ replaced by $D_{\max}^{\eps,T}$ with $\rho \in \cS$.
\end{theorem}

\begin{proof}
Invoking Lemma~\ref{lem:minimax-smoothing} the claim becomes equivalent to
\begin{align}
\sup_{M\geq0\atop \mathrm{Tr}[M\sigma]\leq1}\inf_{\tilde{\rho}\in \cB(\rho)} \tr \left[M\tilde{\rho}\right]\leq2^{D_\alpha(\rho\|\sigma)}\cdot g(\eps)^{\frac{1}{\alpha-1}} h(\eps) ,
\end{align}
where we introduced $g(\eps) = \frac{1}{\eps^2}$ and $h(\eps) = \frac{1}{1-\eps^2}$ for convenience.
That is, for every $M$ with $\tr(M \sigma) \leq 1$ it is sufficient to produce a corresponding $\tilde{\rho}\in\mathcal{B}(\rho)$ that fulfils the bound. For such an $M$ with spectral decomposition $M = \sum_im_i|v_i\rangle\langle v_i|$, and $\alpha > 1$, define
\begin{align}
p_i:=\langle v_i|\rho|v_i\rangle,\quad q_i:=\langle v_i|\sigma|v_i\rangle,\quad \textrm{and} \quad
I:=\left\{i:\;\frac{p_i}{q_i}>2^{D_\alpha(\rho\|\sigma)}\cdot g(\eps)^{\frac{1}{\alpha-1}}\right\} 
\end{align}
and finally $\Pi:=\sum_{i\in I}|v_i\rangle\langle v_i|$.
We now invoke the data-processing inequality for the quantum R\'enyi divergences under the projective measurement $\{|v_i\rangle\!\langle v_i|\}_i$, leading to
\begin{align}
2^{(\alpha-1)\cdot D_\alpha(\rho\|\sigma)}\geq\sum_ip_i^\alpha q_i^{1-\alpha}\geq\sum_{i\in I}p_i\left(\frac{p_i}{q_i}\right)^{\alpha-1}\geq\sum_{i\in I}p_i\left(2^{D_\alpha(\rho\|\sigma)}\cdot g(\eps)^{\frac{1}{\alpha-1}}\right)^{\alpha-1},
\end{align}
where the last inequality follows from the definition of $I$. This implies that
\begin{align}
\label{eq:projbound}
\tr \Pi \rho = \sum_{i\in I} p_i \leq g(\eps)^{-1} = \eps^2 \,.
\end{align}

We are now ready to define our smoothed state, 
\begin{align}
 	\tilde{\rho} := \frac{(1-\Pi)\rho(1-\Pi)}{1 - \tr\Pi\rho} ,
\end{align}
which is normalized if and only if $\rho$ is normalized (and otherwise sub-normalized). By Lemma~\ref{lem:purifiedgentle} we find that $P(\rho, \tilde\rho) = \sqrt{\tr P\rho} \leq \eps$, and thus this state lies in both $\cB_P^{\eps}$ and $\cB_T^{\eps}$. Furthermore,
\begin{align}
(1 - \tr\Pi\rho) \tr M \tilde \rho &= 
\sum_{i\not\in I} p_i \cdot m_i
\leq \sum_{i\not\in I}q_i\cdot m_i \cdot 2^{D_\alpha(\rho\|\sigma)}\cdot g(\eps)^{\frac{1}{\alpha-1}} 
\leq 2^{D_\alpha(\rho\|\sigma)}\cdot g(\eps)^{\frac{1}{\alpha-1}} , \label{eq:mrhotbound}
\end{align}
where the penultimate inequality follows from the definition of $I$ and the last inequality follows from
$\sum_{i} q_i\cdot m_i = \tr M\sigma \leq1$. Finally, we bound $\frac{1}{1-\tr \Pi \rho} \leq \frac{1}{1-\eps^2} = h(\eps)$, concluding the proof.
\end{proof}


\subsection{Relating smooth max-divergence and asymmetric hypothesis testing}

One of the main results in~\cite{tomamichel12} was to establish a close relation between the smooth max-divergence and asymmetric hypothesis testing, which were then used to derive asymptotic bounds. The following relation improves on two bounds established in~\cite[Proposition 13]{tomamichel12} and~\cite[Proposition 4.1]{dupuis12}. 

\begin{theorem}
\label{thm:DmaxDh}
Let $\rho \in \cS$, $\sigma \in \cP$ and $\eps \in (0,1)$ and $\delta \in (0,1-\eps^2)$. It holds that
\begin{align}
 	D_h^{1-\eps}(\rho\|\sigma)  \geq D_{\max}^{\sqrt{\eps},P}(\rho\|\sigma) - \log\frac{1}{1-\eps}  \geq D_h^{1-\eps-\delta}(\rho\|\sigma) - \log\frac{4}{\delta^2}  \,.
\end{align}
\end{theorem}

We note in particular that our new upper bound on $D_{\max}^{\eps,P}(\rho\|\sigma)$ does not depend on the number of distinct eigenvalues of $\sigma$, in contrast to the result in~\cite[Proposition 13]{tomamichel12}. It is also tight in $\eps$, unlike the bound in~\cite[Proposition 4.1]{dupuis12}. This is particularly relevant when attempting to generalize these relations to the infinite-dimensional case.

\begin{proof}
We start with the first inequality. Using Lemma \ref{lem:minimax-smoothing}, we fix an arbitrary $M \geq 0$ such that $\mathrm{Tr}[M\sigma] \leq 1$ and it suffices to construct a state $\tilde\rho \in \cB_P^{\eps}$ such that 
\begin{align}
  \tr M\tilde\rho \leq \frac{1}{\eps'} 2^{D_h^{\eps'}(\rho\|\sigma)} \,,
\end{align} where we set $\eps' = 1 - \eps$ for convenience.
Given the spectral decomposition $M = \sum_i m_i |v_i\rangle\!\langle v_i|$, we define $\mathcal{M}$ as the measurement in the basis $\{ |v_i\rangle \}_i$ and two probability distributions $P := \mathcal{M}(\rho)$ and $Q:= \mathcal{M}(\sigma)$ obtained by measuring $\rho$ and $\sigma$ in this basis. The data-processing inequality for the hypothesis testing divergence and~\eqref{eq:dsdh} yield
\begin{align}
  D_h^{\eps'}(\rho\|\sigma) \geq D_h^{\eps'}(P\|Q) \geq D_s^{\eps'}(P\|Q) =: K \,.
\end{align}
Let us now, for any $\eta > 0$, define the set $I := \{i: P(i) \leq 2^{K+\eta} Q(i)\}$ such that $P(I) > \eps'$ by definition of $D_s^{\eps'}(P\|Q)$. Moreover, let $\Pi:= \sum_{i \not\in I} |v_i\rangle\!\langle v_i|$. We have 
\begin{align}
 \tr \Pi\rho = \tr \mathcal{M}(\Pi)\rho = \tr \Pi \mathcal{M}(\rho)  =  1 - P(I) \leq \eps \,. \label{eq:Ibound}
\end{align}
And, thus, according to Lemma~\ref{lem:purifiedgentle}, we have $P(\rho, \tilde\rho) \leq \sqrt{\eps}$ for the choice
$\tilde\rho := \frac{(1-\Pi)\rho(1-\Pi)}{1 - \tr\Pi\rho}$.
Finally, using that $1 - \tr \Pi\rho > \eps'$ by~\eqref{eq:Ibound}, we find
\begin{align}
\tr M\tilde{\rho} \leq \frac{1}{\eps'}\sum_{i \in I} m_i P(i) \leq \frac{2^{K+\eta}}{\eps'}\cdot \sum_{i\in I} m_i Q(i) = \frac{2^{K+\eta}}{\eps'} \mathrm{Tr}[M\sigma] \leq \frac{2^{D_h^{\eps'}(\rho\|\sigma)+\eta}}{\eps'}.
\end{align}
The first inequality then follows in the limit $\eta \to 0$.

To show the second inequality, we follow the ideas in~\cite{tomamichel12}. Let $\tilde\rho \in \cB_P^{\eps}$ be such that \begin{align}
\tilde \rho \leq 2^{\lambda}\sigma  \qquad \textrm{with} \qquad \lambda = D_{\max}^{\sqrt\eps,P}(\rho\|\sigma) \,,
\end{align}
that is, the state $\tilde\rho$ is an optimal smooth state. Moreover, consider the optimal hypothesis test $0 \leq Q \leq 1$ satisfying
$\tr (1-Q)\rho = 1- \eps - \delta$ and $\log \tr Q\sigma = -D_h^{1-\eps-\delta}(\rho\|\sigma)$.
Then, the data-processing inequality for the fidelity and applied to the positive operator-valued measurement $\{ Q, 1 - Q\}$ yields the following sequence of inequalities:
\begin{align}
  \sqrt{1-\eps} = \sqrt{\bar{F}(\rho,\tilde\rho)} &\leq \sqrt{\tr Q\rho \tr Q\tilde\rho} + \sqrt{\tr (1-Q)\rho \tr (1-Q)\tilde\rho} \\
  &\leq \sqrt{\tr Q\tilde\rho} + \sqrt{\tr (1-Q)\rho} \\
  &\leq \sqrt{2^{\lambda} \tr Q \sigma} + \sqrt{1 - \eps -\delta} \,.
\end{align}
Substituting for $\lambda$ and $\tr Q\sigma$, we thus arrive at the inequality
\begin{align}
\log\left( \sqrt{1-\eps}-\sqrt{1-\eps-\delta} \right)^2 \leq D_{\max}^{\sqrt\eps,P}(\rho\|\sigma) - D_h^{1-\eps-\delta}(\rho\|\sigma) \,.
\end{align}
Further bounding $\sqrt{1-\eps}-\sqrt{1-\eps-\delta} \geq \frac{\delta}{2 \sqrt{1-\eps}}$ yields the desired result.
\end{proof}


\section{Joint smoothing relative to arbitrary states}

Simultanenous smoothing is a question of great interest in quantum Shannon theory, with recent progress such as in \cite{Sen18,drescher13} having new consequences in network scenarios. Here we show simultaneous smoothing for the two marginals of joint quantum system $AB$. In contrast to earlier results on joint smoothing, our technique allows to smooth relative to an arbitrary positive operator, and this operator can in fact be different for the two marginals. If we choose these operators to be identity, our result reduces to the usual case considered in the literature~\cite{drescher13}.
We hope that the approach can lead to more progress on the simultaneous smoothing question.

\begin{theorem}
Let $\rho_{AB} \in \cS(AB)$ with marginals $\rho_A$ and $\rho_B$, and let $\sigma_A \in \cP(A)$, $\sigma_B \in \cP(B)$. For any $\eps, \eps' \in (0,1)$ such that $\eps + \eps' < 1$, there exists a state $\tilde\rho_{AB} \in \cS(AB)$ with $P(\rho_{AB}, \tilde\rho_{AB}) \leq \sqrt{\eps +\eps'}$ such that its marginals $\tilde\rho_A$ and $\tilde\rho_B$ satisfy
\begin{align}
D_{\max}(\tilde\rho_A \| \sigma_A) \leq D_{h}^{1-\eps}(\rho_A \| \sigma_A) + \Delta \label{eq:smoothd}
 \quad \textrm{and} \quad
 D_{\max}(\tilde\rho_B \| \sigma_B) \leq D_{h}^{1-\eps'}(\rho_B \| \sigma_B) + \Delta 
\end{align}
for $\Delta = -\log (1-\eps-\eps')$.
\end{theorem}

\begin{proof}
Let us first confirm that it suffices, for every $\eta > 0$, to construct a normalized state $\tilde\rho_{AB} \in \cB_P^{\sqrt{\delta}}(\rho_{AB})$ for $\delta = \eps+\eps'$ that satisfies the following operator inequalities:
\begin{align}
\tilde\rho_A \leq 2^{\lambda_A} \sigma_A \qquad  \textrm{and} \qquad 
\tilde\rho_B \leq 2^{\lambda_B} \sigma_B , \label{eq:suffcond}
\end{align}
where $\lambda_A = D_{h}^{1-\eps'}(\rho_A \| \sigma_A) + \Delta + \eta$ and $\lambda_B = D_{h}^{1-\eps''}(\rho_B \| \sigma_B) + \Delta + \eta$.
Then, the inequalities in~\eqref{eq:smoothd} are implied since $\eta > 0$ is arbitrarily small. Consider now
\begin{align}
\mathrm{Opt} &:= \inf_{\tilde\rho_{AB}\in \mathcal{B}^\eps(\rho_{AB})} \sup_{0 \leq M_A \leq 1 \atop  0 \leq M_B \leq 1}   \tr M_A(\tilde\rho_A - 2^{\lambda_A} \sigma_A) + \tr M_B(\tilde\rho_B - 2^{\lambda_B} \sigma_B)  \label{eq:minimax2}
\end{align}
where the infimum and supremum can be interchanged using Sion's minimax theorem~\cite{sion58}. Clearly $\mathrm{Opt} \leq 0$ implies the existence of a state satisfying the desiderate in~\eqref{eq:suffcond}. Using the minimax principle on~\eqref{eq:minimax2}, it thus suffices to construct, for every fixed $M_A$ and $M_B$, a $\tilde\rho_{AB} \in \cS(AB)$ with $P(\rho_{AB}, \tilde\rho_{AB}) \leq \sqrt{\delta}$ such that
  $\tr M_A \tilde\rho_A \leq 2^{\lambda_A} \tr M_A \sigma_A$ and $\tr M_B \tilde\rho_B \leq 2^{\lambda_B} \tr M_B \sigma_B$.

The proof now proceeds similarly to the proof of Theorem~\ref{thm:DmaxDh}, where more detail is given. Given the eigenvalue decomposition $M_A = \sum_i m_A(i) |v_i\rangle\!\langle v_i|_A$ of $M_A$, the measurement $\mathcal{M}_A$ in its eigenbasis, and the two probability distributions $P_A = \mathcal{M}_A(\rho_A)$ and $Q_A = \mathcal{M}_A(\sigma_A)$, we find
\begin{align}
  D_h^{1-\eps}(\rho_A\|\sigma_A) \geq D_h^{1-\eps}(P_A\|Q_A) \geq D_s^{1-\eps}(P_A\|Q_A) =: K_A  \,.
\end{align}
We then define the set $I_A = \{i: P(i) \leq 2^{K_A+\eta} Q(i)\}$ such that $P_A(I_A) > 1-\eps$. As a consequence, the projector $\Pi_A := \sum_{i \in I_A} |v_i\rangle\!\langle v_i|_A$ satisfies
\begin{align}
 \tr \Pi_A \rho_A = \tr \mathcal{M}_A(\Pi_A)\rho_A = \tr \Pi_A \mathcal{M_A}(\rho_A)  =  P_A(I_A) \geq 1- \eps. 
\end{align}
The exact same construction for $B$ yields $\Pi_B$ with $\tr \Pi_B \rho_B \geq 1 - \eps'$. Consequently, we establish
\begin{align}
  \tr (1_{AB} - \Pi_A \otimes \Pi_B) \rho_{AB} &= \tr (1_{AB} - \Pi_A \otimes 1_B)\rho_{AB} + \tr  (\Pi_A \otimes 1_B)(1_{AB} - 1_A \otimes \Pi_B) \rho_{AB}  \\
  &\leq  1 - \tr \Pi_A \rho_A + 1 - \tr \Pi_B \rho_B \leq \eps + \eps'  = \delta \,,  \label{eq:Ibound2}
\end{align}
where we used the fact that $\tr_A (P_A \otimes 1_B) X_{AB} \leq X_B$ for every projector $P_A$ and positive operator $X_{AB}$ with marginal $X_B$ (see, e.g.,~\cite[Lemma A.1]{mythesis} for a proof of a more general statement).

Now, we are ready to define the (normalized) smoothed state
\begin{align}
	\tilde\rho_{AB} = \frac{(\Pi_A \otimes \Pi_B) \rho_{AB} (\Pi_A \otimes \Pi_B)}{\tr (\Pi_A \otimes \Pi_B) \rho_{AB}}
\end{align}
such that Lemma~\ref{lem:purifiedgentle} together with~\eqref{eq:Ibound2} yields $P(\rho, \tilde\rho) \leq \sqrt{\delta}$. Moreover,
\begin{align}
  \tr M_A \tilde\rho_A &\leq \frac{1}{1-\delta} \tr (M_A \otimes 1_B) (\Pi_A \otimes \Pi_B) \rho_{AB} (\Pi_A \otimes \Pi_B) \\
  &\leq \frac{1}{1-\delta} \tr M_A \Pi_A \rho_A \Pi_A = \frac{1}{1-\delta} \sum_{i \in I_A} m_A(i) P_A(i) \\
  &\leq  \frac{2^{K_A+\eta}}{1-\delta}  \sum_{i \in I_A} m_A(i) Q_A(i) \,.
\end{align}
Finally, since $ \sum_{i \in I_A} m_A(i) Q_A(i) \leq \tr(M_A \sigma_A)$ and $\frac{2^{K_A+\eta}}{1-\delta} = 2^{\lambda_A}$, the first inequality in~\eqref{eq:suffcond} follows. The analogous argument for $B$ also verifies the second inequality in~\eqref{eq:suffcond}, concluding the proof.
\end{proof}

Using Theorem~\ref{thm:DmaxDh}, we can further replace $D_h^{1-\eps}$ with $D_{\max}^{\sqrt{\eps}}$ and $D_h^{1-\eps'}$ with $D_{\max}^{\sqrt{\eps'}}$ (introducing some small correction), which yields the following corollary.
\begin{corollary}
Let $\rho_{AB} \in \cS(AB)$ with marginals $\rho_A$ and $\rho_B$, and let $\sigma_A \in \cP(A)$, $\sigma_B \in \cP(B)$. For any $\eps, \eps', \delta \in (0,1)$ such that $\eps + \eps' + 2\delta < 1$, there exists a state $\tilde\rho_{AB} \in \cS(AB)$ with $P(\rho_{AB}, \tilde\rho_{AB}) \leq \sqrt{\eps +\eps'+ 2\delta}$ such that its marginals $\tilde\rho_A$ and $\tilde\rho_B$ satisfy
\begin{align}
D_{\max}(\tilde\rho_A \| \sigma_A) \leq D_{\max}^{\sqrt{\eps}}(\rho_A \| \sigma_A) + \Delta
 \quad \textrm{and} \quad
 D_{\max}(\tilde\rho_B \| \sigma_B) \leq D_{\max}^{\sqrt{\eps'}}(\rho_B \| \sigma_B) + \Delta 
\end{align}
for $\Delta = 2 - 2 \log \delta -\log (1-\eps-\eps'-2\delta)$.
\end{corollary}

\paragraph*{Acknowledgements.} 
The work was done when AA was affiliated to the Centre for Quantum Technologies, National University of Singapore. We thank David Sutter for help with the proof of Theorem~\ref{thm:renyi} for normalized trace distance.
AA and RJ were supported by the Singapore Ministry of Education and the National Research Foundation through the ``NRF2017-NRF-ANR004 VanQuTe'' grant. RJ is also supported by VAJRA Faculty Scheme of the Science and Engineering Board (SERB), Department of Science and Technology (DST), Government of India.

\appendix

\section{Proof of Lemma~\ref{lem:purifiedgentle}}
\label{app:gentle}

\begin{proof}[Proof of Lemma~\ref{lem:purifiedgentle}]
We need to verify that $\bar{F}(\rho,\tilde\rho) = 1 - \tr P\rho$ for $\tilde\rho = \frac{(1-P)\rho(1-P)}{1 - \tr P \rho}$. Indeed,
\begin{align}
\sqrt{\bar{F}(\rho,\tilde\rho)} &= \| \sqrt{\rho} \sqrt{\tilde\rho} \|_1 + \sqrt{(1-\tr \rho)(1-\tr \tilde\rho)} \\
 &= \frac{\tr (1 - P) \rho}{\sqrt{1 - \tr P \rho}} + \sqrt{ (1 - \tr\rho) \left(1 - \frac{\tr (1-P)\rho}{1 - \tr P\rho}\right)} 
 = \sqrt{1 - \tr P\rho} 
\end{align}
by a simple computation.
\end{proof}

\section{Proof of Lemma~\ref{lem:minimax-smoothing}}
\label{app:minimax}

\begin{proof}[Proof of Lemma~\ref{lem:minimax-smoothing}]
Recall the definition of the max-divergence, 
$D_{\max}( \rho\|\sigma) = \log \inf_{\rho \leq \lambda \sigma} \lambda$.
We first show the following identity:
\begin{align}
\inf_{ \rho \leq \lambda \sigma}  \lambda = \sup_{X \geq 0} \, \inf_{ \tr X \rho  \leq \lambda \tr X \sigma  }  \lambda \,.
\end{align}

The direction `$\geq$' follows directly from the fact that
\begin{align} 
 \inf_{ \rho \leq \lambda \sigma} \lambda \geq \inf_{ \tr X \rho  \leq \lambda \tr X \sigma  }  \lambda \,, \label{eq:ineq1}
\end{align}
for all $X \geq 0$, since the restriction on $\lambda$ on the right-hand side is less restrictive.

For the direction `$\leq$', we simply need to construct an operator $X \geq 0$ such that the infimum on the right-hand side of~\eqref{eq:ineq1} matches the left-hand side.
We first consider the case where $ \inf_{ \rho \leq \lambda \sigma} \lambda = \infty$, i.e.\ the case where the support of $\rho$ is not contained in the support of $\sigma$. In this case we can choose $X$ to be orthogonal to $\sigma$ but with $\tr X \rho > 0$, such that indeed also $\inf_{ \tr X \tilde \rho  \leq \lambda \tr X \sigma  }  \lambda = \infty$. Otherwise, choose $\lambda^* = \argmin_{ \rho \leq \lambda \sigma}  \lambda$. With $X$ the projector onto the kernel of $\lambda^* \sigma - \rho$, we find
\begin{align}
  \lambda \tr X\sigma - \tr X \rho = (\lambda - \lambda^*) \tr X\sigma +  \lambda^* \tr X\sigma - \tr X \rho =  (\lambda - \lambda^*) \tr X\sigma,
\end{align}
and thus $\inf_{\lambda \geq 0, \tr\, X \tilde \rho  \leq \lambda \tr\, X \sigma  } \lambda = \lambda^*$, as required.
Normalising $X$ such that $\tr X \sigma = 1$ then yields
\begin{align}
D_{\max}( \rho\|\sigma)  =  \log \sup_{X \geq 0 \atop \tr X \sigma = 1}  \tr\, X  \rho  = \log \sup_{X \geq 0 \atop \tr X \sigma \leq 1}  \tr\, X  \rho \, .
\end{align}

And finally, using the definition of $\eps$-smooth max-divergence, we find
\begin{align}
D^{\eps}_{\max}(\rho\|\sigma) = \inf_{\tilde \rho \in \cB^{\eps}(\rho)} \sup \limits_{X \geq 0 \atop \tr X \sigma \leq 1} \log \tr\, X \tilde \rho \, .
\end{align}
Sion's minimax theorem~\cite{sion58} ensures that we can swap the infimum and the supremum since $\cB^{\eps}(\rho)$ and $\{X\geq 0 ,  \tr X \sigma \leq 1 \}$ are convex sets, which completes the proof.
\end{proof}


\bibliographystyle{arxiv_no_month}
\bibliography{library}

\end{document}